\documentclass[letterpaper, 10 pt, conference]{ieeeconf}  

\IEEEoverridecommandlockouts                              
\usepackage{bbm}
\usepackage{xcolor}
\usepackage{soul}
\sethlcolor{yellow}

\usepackage{amsfonts}
\usepackage{amsmath} 
\usepackage{amssymb}  
\usepackage{dsfont}
\usepackage{cases}
\usepackage{hhline}
\usepackage{array}
\usepackage{graphicx}
\usepackage{cite}

\newcommand{\A}{{\mathcal{A}}}

\newcommand{\bbP}{\mathbb{P}}

\newcommand{\R}{\mathbb{R}}

\newcommand{\vs}{\vspace{-1mm}}
\newcommand{\hs}[1][0.5]{\hspace{-#1mm}}

\newcommand{\online}[1]{This proof is omitted for space and appears online \cite{collins2021robust}.}
\renewcommand{\online}[1]{#1}

\newcommand{\nb}[1]{{\hl  {#1}}}
\renewcommand{\nb}[1]{}  

\DeclareMathOperator*{\argmax}{arg\,max}

\newtheorem{definition}{Definition}
\newtheorem{theorem}{Theorem}
\newtheorem{proposition}{Proposition}
\newtheorem{lemma}{Lemma}

\title{\Large \bf
Robust Stochastic Stability in Dynamic and Reactive Enviroments
}

\author{Brandon C. Collins, Lisa Hines, Gia Barboza, and Philip N. Brown
\thanks{This work was supported by the National Science Foundation under Grants \#DEB-2032465 and \#ECCS-2013779.}
\thanks{The authors are with the University of Colorado Colorado Springs,
        CO 80918, USA 
        {\tt\small \{bcollin3,lhines,gbarboza, philip.brown\}@uccs.edu}}
}

\begin{document}

\maketitle
\pagestyle{empty}

\begin{abstract}
The theory of learning in games has extensively studied situations where agents respond dynamically to each other by optimizing a fixed utility function.
However, in many settings of interest, agent utility functions themselves vary as a result of past agent choices.
The ongoing COVID-19 pandemic provides an example: a highly prevalent virus may incentivize individuals to wear masks, but extensive adoption of mask-wearing reduces virus prevalence which in turn reduces individual incentives for mask-wearing.
This paper develops a general framework using probabilistic coupling methods that can be used to derive the stochastically stable states of log-linear learning in certain games which feature such game-environment feedback.
As a case study, we apply this framework to a simple dynamic game-theoretic model of social precautions in an epidemic and give conditions under which maximally-cautious social behavior in this model is stochastically stable.


\end{abstract}
\section{Introduction}
    In social systems and distributed engineered systems, collective behavior is the result of many individuals making intertwined self-interested choices.
    In many cases, the value of a particular choice depends not only on the current choices being made by others, but also on the history of past choices.

	In principle, these socio-environmental feedback loops can be analyzed using techniques from \emph{game theory}, which has a long history of analyzing the society-scale effects of self-interested behavior.
	For instance, game theory has long been used to study the spread of social conventions \cite{young1993evolution} using models such as the graphical coordination game \cite{kearns2013graphical} with the stochastic learning algorithm log-linear learning \cite{marden2012revisiting}.
    However, traditional analysis techniques almost uniformly assume that the game's utility functions are fixed for all time, so that the agents' choices over time can be described by a stationary Markov process.
	However, such analysis fails or becomes unwieldy when utility functions themselves depend on the history of play. 
	
	Analysis techniques for history-dependent games have broad potential applications. 
	For example, in a global pandemic, the individual choice to adopt protective measures (e.g., wearing masks) may be made in response to the behavior of others and the prevalence of the disease.
	In turn, the prevalence of the disease is a function of the history of individual choices to adopt protective measures.
	As another example, game theoretic methods are frequently proposed in the area of distributed control of multiagent systems~\cite{Chandan2019,Collins2020,marden2013distributed,kanakia2016modeling}.
	However, in a distributed control application, agents' actions may directly modify the strategic environment; for instance if a search-and-rescue UAV identifies a disaster victim, that victim may be removed from the list of other UAVs' objectives.
    Other applications that can be modeled by history-dependent games are in machine learning \cite{wang2019evolutionary,garciarena2018evolved,costa2019coegan}
    and biology \cite{tilman2017maintaining,tilman2020evolutionary}.

    Owing in part to the challenges of modeling the complex game-environment feedback inherent to history-dependent games, general results on these games are elusive.
    Recent work has focused on specific learning algorithms and strategic environments, such as zero-sum games under replicator dynamics~\cite{skoulakis2020evolutionary}.
    In \cite{weitz2016oscillating} the authors characterize an oscillating tragedy of the commons effect under certain environmental feedback scenarios.
    
    In this paper, we develop a general framework for analyzing the long-run behavior of binary-action history-dependent games.
    In particular, we study the stochastically stable states of the popular log-linear learning algorithm in such settings.
    We show that if the utility functions of the history-dependent game can be appropriately referenced to the utility functions of a corresponding exact potential game, then the history-dependent game of interest inherits the stochastically stable states of the reference potential game.
    To accomplish this we apply techniques from the theory of probabilistic couplings, and derive a monotone coupling that relates play in the history-dependent game with that in the reference potential game.
    To showcase an application of the framework, we present an epidemic model that intertwines the compartmental SIS disease model with a graphical coordination game.
    Using our analysis framework we provide conditions under which the stochastically stable states may be fully characterized, despite their history-dependence.

\vspace{-0.5mm}
\section{Model}
    \subsection{Game Formulation}
        In this work we consider binary action games.
        Let $N=\{1,2,3,\dots,|N|\}$ denote the player set; player $i\in N$ has action set $A_i=\{0,1\}$.
        The joint action space is then given by $A=\{0,1\}^{|N|}$.
    	We denote an action profile as $a\in A$ and use $a_i$ to denote player $i$'s action.
    	The actions of all other players is then given by $a_{-i}=(a_1,a_2,\dots,a_{i-1},a_{i+1},\dots,a_{|N|})$.
    	We refer to the all $1$ action profile as $\vec{1}=(1)^{|N|}_{i=1}$ and similarly for the all zero profile, $\vec{0}$.
    	Further, let $\Delta(A)$ denote the standard probability simplex over $A$.
    	
    	Let $U_i:A\rightarrow \R$ be player $i$'s utility function.
    	We denote $U=\{U\}_{i\in N}$ as the collection of all players' utility function.
    	
    	A game $g=(N,A,U)$ is an \emph{exact potential game} if there exists a potential function $\phi$ such that\vs
    	\begin{equation}
			U_i(a_i',a_{-1})-U_i(a_i,a_{-i})=\phi(a_i',a_{-1})-\phi(a_i,a_{-i})\vs\vs
	    \end{equation}
	    for any $a\in A$, and $a_i,a_i'\in A_i$.

	    In this work we generalize the above by allowing each history of play to have a unique utility function.
	    %
	    We write $\A_T$ to denote the set of joint action histories of length $T\in\mathbb{N}$,
	    and denote the set of all histories as $\mathcal{A}=\cup_{T\in \mathbb{N}}\mathcal{A}_T$.
	    We write $\alpha\in\mathcal{A}_T$ to refer to a history of action profiles (or \emph{path}) and use superscripts to denote time indices so that $\alpha=(\alpha^1,\dots,\alpha^T)$.
	    We abuse notation and write $\alpha^T$ to denote the last action profile in any path $\alpha$.
	    We also define $A,\mathcal{A}_T$ as partially ordered sets by first defining partial order $\geq_A$, where $a'\geq_A a$ whenever $a,a'\in A$ and $a'_i\geq a_i$ for all $i\in N$, recalling that $a'_i,a_i\in \{0,1\}$.
	    Using this we define partial order $\geq_{\mathcal{A}_T}$ as $\bar{\alpha}\geq_{\mathcal{A}_T}\alpha$ whenever $\alpha,\bar{\alpha}\in \mathcal{A}_T$ and $\bar{\alpha}^t\geq_{A}\alpha^t$ for all $t\in\{1,2,...,T\}$.
	    
	    To model history-dependent utility functions, let $U^\alpha_i:A\rightarrow \mathbb{R}$, where this utility function is not only specific to player $i$ but also to the history $\alpha$.
    	Let $U^\alpha=(U^\alpha_1,U^\alpha_2,...,U^\alpha_{|N|})$ denote each player's utility function given history $\alpha$
    	and let $U^\mathcal{A}=\{U^\alpha\mid\alpha\in\mathcal{A}\}$ be the set of utility functions across all paths.
        We denote a \emph{history-dependent game} as tuple $(N,A,U^{\mathcal{A}})$ and let $\mathcal{G}^{\mathcal{A}}$ be the set of all such tuples.
        We now present a class of games that combines potential games and history dependence.

    	\begin{definition}\label{def:binary action history var}
    		We call a tuple $g=(N,A,U^\mathcal{A})\in \mathcal{G}^{\mathcal{A}}$ an \textit{aligned history-dependent game} if there exists an exact potential game $\hat{g}=(N,A,\hat{U})$ with potential function $\hat{\phi}$ such that:
    		\begin{enumerate}
    			\item $\{\vec{1}\}=\argmax_{z\in A}{\hat{\phi}(z)}$
    			\item $U^\alpha_i(1,\alpha^T_{-i})\geq\hat{U}_i(1,a_{-i})$
    			\item $\hat{U_i}(0,a_{-i})\geq U^{\alpha}_i(0,\alpha^T_{-i})$
    		\end{enumerate}
            for any $\alpha\in\mathcal{A}$, $a,a'\in A$, $T\in\mathbb{N}$ such that $\alpha^T_{-i}\geq_{A_{-i}} a_{-i}$ and $a,a'$ vary by only a unilateral deviation.         
            For convenience, we denote ordering $\geq_{A_{-i}}$ over $A_{-i}=\{0,1\}^{|N|-1}$ equivalently to $\geq_A$.
        \end{definition}


    \subsection{Learning in Games}
    
    	In this work we focus on the learning algorithm log-linear learning, which is a discrete time asynchronous learning algorithm \cite{alos2010logit,marden2012revisiting}.
    	That is, for game $g\in\mathcal{G}^\mathcal{A}$ at each time step log-linear learning selects a single agent uniformly at random to update their action.
        \vs
        \begin{equation}\label{eqn:bbP}
        	\bbP^\alpha_i(a_i)=\frac{\exp({\frac{1}{\tau}U^\alpha_i(a_i,\alpha^T_{-i})})}{\sum_{a_i'\in A_i}\exp({\frac{1}{\tau}U^\alpha_i(a_i,\alpha^T_{-i})})},
        	\vs
        \end{equation}
    	where $\exp(x):=e^x$, and $\tau$ is the \emph{temperature}, a parameter which governs the rationality of agents.
    	As $\tau\to0$ agents best respond with high probability, and as $\tau\to\infty$ agents choose actions uniformly at random.
       Note that we take the last action profile in the history $\alpha^T$ as the behavior of the other agents.
        The probability that $\alpha\in\mathcal{A}_T$ transitions to $a'\in A$ under log-linear learning in a single transition is \vs
        \begin{equation}\label{eqn:P}
        	P^\alpha(a')=
        	\begin{cases}
        		\frac{1}{|N|}\sum_{j\in N}\bbP^\alpha_j(a'_j) & \alpha^T=a' \\
        		\frac{1}{|N|}\bbP_i^\alpha(a'_i) & \alpha^T_i\neq a'_i,\alpha^T_{-i}= a'_{-i} \\
        		0 & \mbox{else}.
        	\end{cases}\vs
        \end{equation}
        This can be interpreted as the probability that given history $\alpha$ the next action profile $\alpha^{T+1}=a'$.
        	 
    	We say $a\in A$ is \emph{strictly} stochastically stable if the following definition holds, due to~\cite{brown2019}.
    	For any $\epsilon>0$ there exists $\mathcal{T}>0,T<\infty$ such that\vspace{-1.5mm}
		\begin{equation}\label{eqn:SS to g hat}
			\mbox{Pr}(s(t;\tau,\pi,g)=\vec{1})>1-\epsilon \mbox{ whenever }t>T,\tau<\mathcal{T}\vspace{-1.5mm}
		\end{equation}
	    where $s(\cdot)$ is a random variable representing the action profile at time $t$ under log-linear learning, given temperature $\tau$, initial distribution $\pi\in\Delta(A)$ and game $g$.
	    
	    Exact potential games under log-linear learning may be analyzed using a theory of \emph{resistance trees} \cite{young1993evolution,alos2010logit,pradelski2012learning,marden2012revisiting} to relate potential function maximizers to stochastic stability.
	    However, this analysis depends on the fact that log-linear learning induces an ergodic Markov process on the action profiles of any exact potential game, and 
	    it is unclear how to apply resistance tree techniques generally on history-dependent games to show stochastic stability.
\vspace{-0.5mm}
\section{Main Contribution}
    We now present our main result, giving that $\vec{1}$ is stochastically stable in aligned history-dependent games
	\begin{theorem} \label{thm:SS}
		If $g\in \mathcal{G}^{\mathcal{A}}$ is an aligned history-dependent game then $\vec{1}$ is strictly stochastically stable in $g$ under log-linear learning.
	\end{theorem}
	
	The proof of Theorem~\ref{thm:SS} proceeds using Lemma~\ref{thm:proof of stoch dom+}, which we present here and prove in Section~\ref{sec:proofs}.
	The interpretation of this lemma is that for an aligned history dependent game $g$, the probability at any time step that $g$ is in the $\vec{1}$ action profile is lower bounded by the probability its associated exact potential game $\hat{g}$ is in the $\vec{1}$ profile.
	 \begin{lemma}\label{thm:proof of stoch dom+}
        If $g\in\mathcal{G}^\mathcal{A}$ is an aligned history-dependent game with associated exact potential game $\hat{g}$ then  ${\rm Pr}(s(T;\tau,\pi,g)=\vec{1})\geq{\rm Pr}(s(T;\tau,\pi,\hat{g})=\vec{1})$ for any temperature $\tau>0,\pi\in\Delta(A),T\in\mathbb{N}$.
    \end{lemma}
    
	The proof of Lemma~\ref{thm:proof of stoch dom+} is technically involved and depends on our novel monotone coupling framework which we present in Section~\ref{sec:proofs}.
	Using this result we now present a straightforward proof of Theorem~\ref{thm:SS}.
	
	\subsubsection*{Proof of Theorem~\ref{thm:SS}}
	    Let $g\in\mathcal{G}^{\mathcal{A}}$ be an aligned history-dependent game and $\hat{g}$ be its associated exact potential game.
    	It is well-known~\cite{alos2010logit} that in an exact potential game $a\in A$ is a stochastically stable state under log-linear learning if\vs\vs
    	\begin{equation}
    	    a\in\argmax_{a'\in A}\phi(a').\vs\vs
    	\end{equation}
        Therefore, because $\vec{1}$ is the lone maximizer of $\hat{\phi}$, it is strictly stochastically stable.
	    We apply Lemma~\ref{thm:proof of stoch dom+} directly to the definition of strict stochastic stability in~\eqref{eqn:SS to g hat}.
		For any $\epsilon>0$ there exists $\mathcal{T}>0,T<\infty$ such that\vs
		\begin{equation}\vs
			\mbox{Pr}(s(t;\tau,\pi,g)=\vec{1})\geq\mbox{Pr}(s(t;\tau,\pi,\hat{g})=\vec{1})>1-\epsilon\; \vs
		\end{equation}
		for all $t>T,\tau<\mathcal{T}$, yielding stochastic stability of $\vec{1}$ in game $g$.
    \hfill$\blacksquare$

	
	
\vspace{-1mm}
\section{A Social Distancing Example}
    To highlight Theorem~\ref{thm:SS}'s ability to analyze stochastic stability of history-dependent games, we exhibit a case study on a simple model of epidemics.
    One challenge of epidemic modeling is to account for the interplay between epidemic severity and the voluntary adoption of preventative social conventions.
    For example, in the absence of a epidemic people prefer not to wear masks; however, in a widespread epidemic people may prefer to take preventative measures.
    To model this phenomenon we intertwine the SIS compartmental epidemic model and the graphical coordination game (GCG) which models the spread and adoption of the relevant preventative social conventions; we term this model SISGCG.
    The fraction of individuals in the society susceptible to infection is described by the nonlinear differential equation\vs\vs
	\begin{equation}\label{eqn:SISGCG}
		\dot{s}=(1-s)(\gamma-\beta(t) s),\vs\vs
	\end{equation}
	where $\gamma>0$ is the curing rate and $\beta(t)>0$ is a rate of infection which depends on agent actions.
	The action $1$ represents a ``safe convention'' action in which a player is acting to reduce contagion; the action $0$ represents conventions ignoring the pandemic.
	These actions are associated with infection coefficients $0<\beta_1<\beta_0$, respectively.
	Accordingly, $\beta(t)$ is simply the average infection rate of all individuals, given their choices:\vspace{-1.5mm}
	\begin{equation}\label{eqn:betat}
		\beta(t)=\frac{1}{|N|}\sum_{i\in N}a^t_i\beta_1+(1-a^t_i)\beta_0,\vspace{-1.5mm}
	\end{equation}
	where $a_i^t$ is the action selected by player $i$ at time $t\in\mathbb{N}$.
    Actions are selected by agents in $N$ dynamically on undirected graph $G=(N,E)$ according to log-linear learning~\eqref{eqn:bbP}. 
	The utility of agent $i$ at time $t$ is given by\vs
	\begin{equation}\label{eqn:uSISGCG}
		\bar{U}_i^\alpha(a_i^t,a_{-i}^t)=a_i|\mathcal{N}_i(1)|(q+I(t))+(1-a_i)|\mathcal{N}_i(0)|,
		\vs
	\end{equation}
	where $\mathcal{N}_i(x)=\{j\in N\mid (i,j)\in E, a_j=x\}$ is the set of $i$'s neighbors who play $x\in \{0,1\}=A_i$, the fraction of infected individuals is given by $I(t) := 1-s(t)$, and $q\in (0,1]$ represents the agent's willingness to practice safe conventions in the absence of an epidemic.

	\begin{proposition}\label{thm:positive invariant}
	    If $s(0)\in[0,1)$, then if $s(t)$ is a solution of~\eqref{eqn:SISGCG} with $\beta(t)$ given by~\eqref{eqn:betat}, there exists a $\bar{t}$ such that $s(t)\leq \gamma/\beta_1$ for all $t\geq \bar{t}$ almost surely.
	\end{proposition}
    \online{
	\begin{proof}
	    We write $s_1^*:= \gamma/\beta_1.$
	    Note that if $s(t)\geq s_1^*$, then because $\beta(t)\geq\beta_1$, we have that $\dot{s}\leq0$ by~\eqref{eqn:SISGCG}, and that this inequality is strict whenever $s(t)> s_1^*$.
	    Thus, the set $[0,s_1^*]$ is positively invariant for the hybrid nonlinear dynamics given in~\eqref{eqn:SISGCG}.
	    
	    To see that $s(t)$ eventually enters $[0,s_1^*]$ almost surely, consider the event that $s(t)>s_1^*$ for all $t$.
	    Since $s_1^*$ is asymptotically stable when $\beta(t)\equiv \beta_1$ and for any action profile $a\neq\vec{1}$ that its associated $\beta(t)>\beta_1$, it follows that the event that $\beta(t)\equiv \beta_1$ for all $t$ is the same event as $s(0)>s_1^*$ and $s(t)>s_1^*$ for all $t$.
	    However, it can be seen that the log-linear learning~\eqref{eqn:P} action update probabilities define a stochastic process which visits every action profile in $A$ infinitely often.
	    That is, the probability that $\beta(t)\equiv \beta_1$ is $0$, and thus there must exist a $\bar{t}$ such that $s(t)\leq s_1^*$ for all $t\geq\bar{t}$ almost surely.
	\end{proof}
	}
	
    It can be seen from \eqref{eqn:uSISGCG} that SISGCG can be represented by a history-dependent game, as the utility function depends on the history of play, so
    our Theorem~\ref{thm:SS} allows us to reference SISGCG to a related exact potential game and deduce conditions guaranteeing that $\vec{1}$ is strictly stochastically stable.
	\begin{proposition}\label{thm:sisgcg is ss}
		Let $g^S$ be an instance of SISGCG.
		If $\beta_1/\gamma>1$, $q+\gamma/\beta_1>1$ and $I(0)>0$ then $\vec{1}$ is stochastically stable in $g$.
	\end{proposition}
	\begin{proof}
        Let the SISGCG model be denoted by $g^S$, which played on graph $G=(N,E)$ with $q+\gamma/\beta_1>1$ and $I(0)>0$, and we consider $g^S$ as played after time $\bar{t}$ as given by Proposition~\ref{thm:positive invariant}.
		Game $g^S$ is a history-dependent game since~\eqref{eqn:uSISGCG} depends on $I(t)$, which is itself a function of the history $\alpha$.
		Thus we have $g^S=(N,A,\bar{U})\in\mathcal{G}^{\mathcal{A}}$ where we let $\bar{U}=\{\bar{U}^\alpha\mid \alpha\in\mathcal{A}\}$.
		
		Now we let $\hat{g}^S=(N,A,\hat{U}^S)$ be a graphical coordination game played on graph $G$, where the utility function $\hat{U}^S$ is given by~\eqref{eqn:uSISGCG} with $I(t)= \gamma/\beta_1$.
		Standard results give that $\hat{g}^S$ is an exact potential game and that $\vec{1}$ is its lone potential function maximizer~\cite{young1993evolution}.
	    %
    	
    	We now use $\hat{g}^S$ to show $g^S$ is an aligned history-dependent game.
    	Now we verify $U_i^\alpha(1,\alpha^T_{-i})\geq \hat{U}^S_i(1,a_{-1})$ anytime $\alpha^T_{-i}\geq_{A_{-i}}a_{-i}$, $t>\bar{t}$.
    	This can be rewritten for $t>\bar{t}$ as\vs\vs
    	\begin{equation}
    	    \sum_{j\in\mathcal{N}_i(1;\alpha^T_{-i})}q+I(t)\geq\sum_{j\in\mathcal{N}_i(1;a_{-i})}q+\gamma/\beta_1\vs\vs
    	\end{equation}
    	where $\mathcal{N}_i(1;a_{-i})$ denotes the neighbors of $i$ who are playing $1$ given profile $a$.
    	This expression holds because $\alpha^T_{-i}\geq_{A_{-i}} a_{-i}\Rightarrow |N_i(1;\alpha^T_{-i})|\geq |N_i(1;a_{-i})|$ and by Proposition~\ref{thm:positive invariant}. 
    	An argument with the same structure holds for $\bar{U}_i^\alpha(0,\alpha^T_{-i})\leq \hat{U}^S_i(0,a_{-1})$.
    	Thus $g^S$ is an aligned history-dependent game, and Theorem~\ref{thm:SS} gives $\vec{1}$ is strictly stochastically stable. 
		%
		%
				%
		%
	\end{proof}

\section{Proof of Lemma~\ref{thm:proof of stoch dom+}}\label{sec:proofs}
\subsection{A Primer on Monotone Couplings}
    We begin with the definition of a monotone coupling, the core analytical device for our paper.
	\begin{definition}\label{def:monotonic coupling}
			Let $X$ be a countable set with partial ordering $\leq_X$ and $p_1,p_2$ be probability measures on measure space $(X,\mathcal{F})$. 
			Then a \emph{monotone coupling} of $p_1,p_2$ is a probability measure $p$ on $(X^2,\mathcal{F}^2)$ satisfying the following for all $x,y\in X$ \vs\vs
		\begin{equation}\label{eqn:monotonic coupling}
				\sum_{x\leq_{X}y'} p(x,y')=p_2(y') \mbox{ and }\sum_{y\geq_Xx'}p(x',y)=p_1(x').
		\end{equation}
	\end{definition}
	
	A monotone coupling is a useful tool for analyzing the component probability measures $p_1$ and $p_2$.
	In particular the following property holds in general for monotone couplings.
    \begin{proposition}[Paarporn et al.,~\cite{paarporn2017}]\label{thm:increasing functions in expectation}
        Let $p_1,p_2$ be probability measures on $(X,\mathcal{F})$.
        If $p$ is a monotone coupling of $p_1,p_2$ then for any increasing random variable $Z:X\rightarrow \mathbb{Z}_+$ we have\vs\vs\vs\vs
        \begin{equation}
            \mathbb{E}_{p_1}(Z)-\mathbb{E}_{p_2}(Z)=\sum_{\eta=0}^\infty p(Z^c_\eta,Z_\eta)\vs\vs\vs
        \end{equation}
        where $Z_\eta=\{a\mid Z(a)>\eta\}$.
    \end{proposition}
    Here we denote $Z^c:=X\setminus Z$ to be the complement set of $Z\subseteq X$.
    The proof is given in \cite[Proposition 1]{paarporn2017}.
	
\subsection{Notation Required for Proofs}
    Taking $\hat{g}=(N,A,U)$, we give equations analogous to \eqref{eqn:bbP}, \eqref{eqn:P} that give the transition probabilities for $\hat{g}$ under log-linear learning.
    In particular, if agent $i$ is selected to update her action then she will do so with probabilities given by\vspace{-1.75mm}
	\begin{equation}\label{eqn:bbP hat}
		\hat{\bbP}^a_i(a_i)=\frac{\exp({\frac{1}{\tau}U_i(a_i,a_{-i})})}{\sum_{a_i'\in A_i}\exp({\frac{1}{\tau}U_i(a'_i,a_{-i})})}
		\vspace{-1.75mm}
	\end{equation}
	Building on \eqref{eqn:bbP hat}, we define the probability that action profile $a$ transitions to $a'$ under log-linear learning in a single transition as \vs\vs
	\begin{equation}\label{eqn:P hat}
		\hat{P}^a(a')=
		\begin{cases}
			\frac{1}{|N|}\sum_{j\in N}\hat{\bbP}^a_j(a_j) & a=a' \\
			\frac{1}{|N|}\hat{\bbP}_i^a(a'_i) & a_i\neq a'_i,a_{-i}= a'_{-i} \\
			0 & \mbox{else}
		\end{cases}\vs\vs
	\end{equation}
	for some $i\in N$ and $a,a'\in A$.
	Additionally, we define the probability that path $\alpha\in\mathcal{A}_T$ occurs with initial distribution $\pi \in \Delta(A)$ as\vs\vs
	\begin{equation}\label{eqn:Ppi hat}
		\hat{P}_\pi(\alpha)=\pi(\alpha^1)\prod_{t=1}^{T-1} \hat{P}^{\alpha^{t}}(\alpha^{t+1})\vs\vs
	\end{equation}	
	noting that $\pi(\alpha^1)$ denotes the probability of $\alpha^1$ in initial distribution $\pi$.

	Correspondingly, the probability that path $\alpha\in\mathcal{A}_T$ occurs with initial distribution $\pi \in \Delta(A)$ on $g\in\mathcal{G^{\mathcal{A}}}$ is\vs\vs
	\begin{equation}\label{eqn:Ppi}
		P_\pi(\alpha)=\pi(\alpha^1)\prod_{t=1}^{T-1} P^{\alpha^{\leq t}}(\alpha^{t+1})\vs\vs
	\end{equation}	
	where we use $\alpha^{\leq t}\in\mathcal{A}_t$ to mean history $\alpha$ until time $t\in\{1,2,3,\dots,T\}$. 
	
	We now present a result connecting the utility conditions of aligned history-varying potential games to \eqref{eqn:bbP hat} and \eqref{eqn:bbP}.
	\begin{lemma}\label{thm:the bbP property}
	    Let $g=(N,A,U^\mathcal{A})\in\mathcal{G}^{\mathcal{A}}$, $\hat{g}=(N,A,\hat{U})$ be an exact potential game and let $i\in N,a\in A,\alpha\in\mathcal{A}$ such that $\alpha^T_{-i}\geq_{A_{-i}} a_{-i}$.
	    If $U^\alpha_i(1,\alpha^T_{-i})\geq \hat{U}_i(1,a_{-i})$ and $\hat{U}_i(0,a_{-i})\geq U^\alpha_i(0,\alpha^T_{-i})$ then $\bbP_i^\alpha(1)\geq\hat{\bbP}_i^a(1)$.
	\end{lemma}
	\online{
	\begin{proof}
	    Let $g=(N,A,U^\mathcal{A})\in\mathcal{G}^{\mathcal{A}}$, $\hat{g}=(N,A,\hat{U})\in\mathcal{G}$ and let $i\in N,a\in A,\alpha\in\mathcal{A}$ be such that $\alpha^T_{-i}\geq_{A_{-i}} a_{-i}$.
	    Further let $U^\alpha_i(1,\alpha^T_{-i})\geq \hat{U}_i(1,a_{-i})$ and $\hat{U}_i(0,a_{-i})\geq U^\alpha_i(0,\alpha^T_{-i})$.
    	Recalling $\tau>0$, we begin by considering $\bbP_i^\alpha$
	    \begin{equation}\label{eqn:bbP property proof}
	    \begin{aligned}
            \bbP_i^\alpha(1)&=
            \frac{e^{\frac{1}{\tau}U^\alpha_i(1,\alpha^T_{-i})}}
            {e^{\frac{1}{\tau}U^\alpha_i(1,\alpha^T_{-i})}+e^{\frac{1}{\tau}U^\alpha_i(0,\alpha^T_{-i})}}
            \\
            &\geq
            \frac{e^{\frac{1}{\tau}\hat{U}_i(1,a_{-i})}}
            {e^{\frac{1}{\tau}\hat{U}_i(1,a_{-i})}+e^{\frac{1}{\tau}\hat{U}_i(0,a_{-i})}}
            =\hat{\bbP}_i^a(1).
	    \end{aligned}
	    \end{equation}
	    To see the inequality, it suffices to apply the hypothesis to the fact that $e^x$ and $l(x)=\frac{e^x}{e^x+c}$ are both increasing in $x$ for $c>0$.
	    Thus $\bbP_i^\alpha(1)\geq\hat{\bbP}_i^a(1)$ holds as desired.
	\end{proof}
	}


    Our framework requires a careful partitioning of the action space corresponding to different types of agent action deviations.
	Let $f:A\rightarrow 2^{A}$ be defined as  $f(a)=\{a'\in A \mid a_i\neq a_i',a_{-i}=a_{-i}' \mbox{ for } i\in N\}$ be the set of action profiles reachable from $a$ via exactly one unilateral deviation.
	For $a,a'\in A$ let\vs\vs
	\begin{equation}
			g(a,a')=
			\begin{cases}
				i & a_i\neq a_i' \\
				0 & a=a' 	
			\end{cases}\vs\vs
	\end{equation}
	indicate which agent unilaterally deviated their action between action profiles $a,a'$. 

	Now, let $a,a'\in A$ where $a'\geq_A a$.
	We denote several disjoint subsets of $f(a)$:
	\begin{enumerate}
		\item $r(a)=\{z\in f(a) \mid a_{g(a,z)}=1\}$,
		\item $q(a,a')=\{z\in f(a)\mid  z\leq_{A} a'\}\setminus r(a)$, and
		\item $s(a,a')=f(a)\setminus (q(a,a')\cup r(a))$.
	\end{enumerate}
	These sets can be interpreted in the following way. 
	The set $r(a)$ is the set of action profiles which decreased with respect to $\geq_A$ and $q(\cdot)$, $s(\cdot)$ both increased.
	Between $q(\cdot)$ and $s(\cdot)$, $q(\cdot)$'s action profiles remain less than $a'$ and $s(\cdot)$'s profiles are greater then or incomparable to $a'$.
	We now present three more analogous sets that are disjoint subsets of $f(a')$:
	\begin{enumerate}
		\item $R(a')=\{z\in f(a') \mid a'_{g(a',z)}=0\}$,
		\item $Q(a,a')=\{z\in f(a')\mid z\geq_{A} a\}\setminus R(a')$, and
		\item $S(a,a')=f(a')\setminus (Q(a,a')\cup R(a))$.
	\end{enumerate}
	The interpretation of these sets are reversed relative to $r(\cdot)$, $q(\cdot)$ and $s(\cdot)$.
	
	We now highlight some useful features of these sets.
    It is evident that $q(\cdot),r(\cdot),s(\cdot)$ are a disjoint partition of $f(a)$, and that $Q(\cdot),R(\cdot),S(\cdot)$ are a disjoint partition of $f(a')$.
	For any $a,a'$, $a'\geq_A a$ we relate these sets by a function $b^{a,a'}:f(a)\rightarrow f(a')$.
	To evaluate $b^{a,a'}(\bar{a})$,  identify the agent who deviated their action between $a,\bar{a}$ and then deviate that agent's action in $a'$. 
	Formally, $b^{a,a'}(\bar{a})=(\neg a'_{g(a,\bar{a})},a'_{-g(a,\bar{a})})$ where we define $\neg a_i \in \{0,1\}\setminus \{a_i\}$ for $a_i\in A_i=\{0,1\}$.
	In particular, this function relates the disjoint subsets of $f(a),f(a')$ according to the following lemma.
	\begin{lemma}\label{thm:bijection}
	    If $a,a'\in A$ and $a\leq_{A} a'$, then the following statements hold:
	    \begin{enumerate}
	        \item $b^{a,a'}:r(a)\rightarrow S(a,a')$ is a bijection,
	        \item $b^{a,a'}:s(a,a')\rightarrow R(a')$ is a bijection, and
	        \item $b^{a,a'}:q(a,a')\rightarrow Q(a,a')$ is a bijection.
	    \end{enumerate}
	\end{lemma}
\online{
\begin{proof}
    Let $a,a'\in A$ such that $a'\geq_A a$.
    We proceed by proving $b^{a,a'}:r(a)\to S(a')$ is a bijection; the other bijection statements are proved similarly.

	We begin by proving injectiveness, that is $b^{a,a'}(z)=b^{a,a'}(z')\implies z=z'$ for $z,z'\in r(a)$.
	Observe $g(a,z)=g(a',b^{a,a'}(z))=g(a',b^{a,a'}(z'))=g(a,z')$ where the first and third inequalities follow by definition of $b^{a,a'}$ and the middle by hypothesis.
	Injectiveness follows from $g(a,z)=g(a,z')$ meaning $a,z$ and $a,z'$ differ by the same agent's unilateral deviation. 
	In that context, the possible actions agent $g(a,z)$ is given by $A_{g(a,z)}\setminus \{a_{g(a,z)}\}$ which is a singleton by the binary action property, leaving only one possible state $a$ could transition to in $r(a)$ via a unilateral deviation.
	Thus $z=z'$ as desired.
	
	Next we show surjection, that is for any $z'\in S(a,a')$ there exists a $z\in r(a)$ such that $b^{a,a'}(z)=z'$, for $a,a'\in A$ and $a\leq_{A}a'$.
	By definition of $S(a,a')$, $z'\ngeq a'$, but as $z'\in f(a')$ $z',a'$ differ by only a single unilateral deviation by some agent $i$.
	By partial ordering $\leq_{A}$ we may infer $a'_i=1,z'_i=0$ else $z'\ngeq a'$ would be violated.
	Further, we may infer $a=1$ as suppose $a=0$, then $z'\in Q(a,a')$, giving a contradiction to the definition of $z'$.
	It is easy to see by definition of $r(a)$ that $a_i=1\implies z\in r(a)$ satisfying $g(a,z)=g(a',z')$ as $z_i\neq a_i$ but $z_{-i}=a_{-i}$ by $z\in f(a)$.
	Note $g(a,z)=g(a',z')$ is always satisfied when $b^{a,a'}(z)=z'$ by definition of the function.
\end{proof}
}
	

\subsection{The One-Step Couplings}
    To prove Lemma~\ref{thm:proof of stoch dom+} and obtain Theorem~\ref{thm:SS}, we construct a monotone coupling $\nu^{\hat{g}}_\pi$ between measures $P_\pi,\hat{P}_\pi$.
    We first construct a family of monotone couplings for each one-step transition (Theorem~\ref{thm:nu}),
    which we apply to show the coupling over histories (Theorem~\ref{thm:paths}).
	\begin{theorem}
		\label{thm:nu}					

		Let $g\in \mathcal{G}^{\mathcal{A}}$ denote an aligned history-dependent game and $\hat{g}$ be its associated exact potential game.
	    Then a monotone coupling exists between $\hat{P}^a$ and $P^\alpha$ for any $\alpha\in\mathcal{A},a\in A$ whenever $a\leq_A\alpha^T$. 
		This monotone coupling $\nu^{a,\alpha}:A^2\rightarrow [0,1]$ is given in (\ref{eqn:nu}) in Figure~\ref{fig:onestep coupling}.
	\end{theorem}

		\begin{proof}
			Let $a\in A,\alpha\in\mathcal{A}$ such that $a\leq_A \alpha^T$ and let $g\in\mathcal{G}^{\mathcal{A}}$ be an aligned history-dependent game where $\hat{g}$ is its associated exact potential game.
			To verify $\nu^{a,\alpha}$ is a monotone coupling we must show the following conditions from Definition~\ref{def:monotonic coupling} for any $\bar{a},\bar{a}'\in A$:
			\begin{enumerate}
				\item \label{cond:well defined} $\nu^{a,\alpha}$ is a well-defined probability measure, 
				\item \label{cond:lesser sum} $\sum\limits_{z'\geq_{A}\bar{a}}\nu^{a,\alpha}(\bar{a},z')=\hat{P}^{a}(\bar{a})$, and 
				\item \label{cond:greater sum} $\sum\limits_{z\leq_{A}\bar{a}'}\nu^{a,\alpha}(z,\bar{a}')=P^{\alpha}(\bar{a}')$.
			\end{enumerate}

			We begin by verifying Condition \ref{cond:lesser sum}.
			We consider cases $\bar{a}\notin (f(a)\cup \{a\}),$ $\bar{a}\in q,\bar{a}\in r,\bar{a}\in s$ and $\bar{a}=a$ separately.
			We use the notational convention that $q,s,Q,S$ are assumed to take arguments $(a,\alpha^T)$ and $r,R$ take the argument $a,\alpha^T$ respectively.
			The first case represents any $\bar{a}$ that cannot be achieved in a single unilateral deviation from $a$.
			Trivially, this gives that $\hat{P}^a(\bar{a})=0$, and thus all pairs of $\bar{a},z'$ must satisfy $\nu^{a,\alpha}(\bar{a},z')=0$.
			This holds as all parts of \eqref{eqn:nu} require $\bar{a}\in (f(a)\cup \{a\})$ except \eqref{eqn:nu:else}, which has the desired property.
			
			We now consider the second case that $\bar{a}\in q$.
			Note that only (\ref{eqn:nu:q}) satisfies this condition, so\vs\vs 
			\begin{equation}
				\begin{aligned}
				\sum_{z'\geq_{A}\bar{a}'}\nu^{a,\alpha}(\bar{a},z')&=\nu^{a,\alpha}(\bar{a},\alpha^T) \\ 
				&=\hat{\bbP}^a_{g(a,\bar{a})}(\bar{a}_{g(a,\bar{a})})/|N| =\hat{P}^{a}(\bar{a})
				\end{aligned}\vs\vs			
			\end{equation}
		as desired.
			
			Next we consider $\bar{a}\in r$ which satisfies \eqref{eqn:nu:r}, \eqref{eqn:nu:rS} uniquely since $b^{a,\alpha^T}$ is a bijection by Lemma~\ref{thm:bijection}.
			Thus \vs\vs
			\begin{equation}
				\begin{aligned}
					\sum_{z'\geq_{A}\bar{a}'}\nu^{a,\alpha}(\bar{a},z')&=
					\frac{1}{|N|}\big(\hat{\bbP}^a_{g(a,\bar{a})}(0)\\
					&\qquad-\bbP^\alpha_{g(a,\bar{a})}(0)
					+\bbP^\alpha_{g(\alpha^T,\bar{a}')}(0)\big)\\
					&=\frac{1}{|N|}\hat{\bbP}^a_{g(a,\bar{a})}(0)
					=\hat{P}^{a}(\bar{a})
				\end{aligned}\vs\vs
			\end{equation}
			where the second equality follows as $g(a,\bar{a})=g(\alpha^T,\bar{a}')$ by definition of $b^{a,\alpha^T}$.
			The third equality follows as $\bar{a}\in{r}\implies \bar{a}_{g(a,\bar{a})}=0$.
			
			Considering $\bar{a}\in s$, we find only \eqref{eqn:nu:sR} applies, thus for\vs\vs
			\begin{equation}
				\sum_{z'\geq_{A}\bar{a}'}\nu^{a,\alpha}(\bar{a},z')
				=\frac{1}{|N|}\hat{\bbP}^a_{g(a,\bar{a})}(1)
				=\hat{P}^{a}(\bar{a})\vs\vs
			\end{equation}	
			where $\bar{a}\in s\implies\bar{a}_{g(a,\bar{a})}=1$ or else $\bar{a}$ would be in $q$.
			
			The final case for Condition~\ref{cond:lesser sum} is $\bar{a}=a$.
			we find cases \eqref{eqn:nu:R}, \eqref{eqn:nu:Q}, and \eqref{eqn:nu:corner} apply yielding:\vs
			\begin{equation}
			\begin{aligned}
				\sum_{z'\geq_{A}\bar{a}'}&\nu^{a,\alpha}(\bar{a},z')
				=\frac{1}{|N|}\bigg(|N|-\sum_{z\in q\cup r}\hat{\bbP}^a_{g(a,z)}(z_{g(a,z)})\\
				&\qquad-\sum_{z'\in R}\hat{\bbP}^a_{g(\alpha^T,z')}(1)\bigg)\\
				&=\frac{1}{|N|}\sum_{z\in f(a)}\left(1-\hat{\bbP}^a_{g(a,z)}(z_{g(a,z)})\right)
				=\hat{P}^{a}(\bar{a})	\vs		
				\end{aligned}
			\end{equation}
			where the first equality follows as sums over $Q\cup R$ are equivalent to the sums over $Q$ and $R$ as $Q,R$ are disjoint, and that $z'\in R\Leftrightarrow z'_{g(\alpha^T,z')}=1$ by definition of $R$.
			The second equality follows as the $R$ sum is equivalent to one over $s$ by bijection $b^{a,\alpha^T}$, and then we may combine it with the sum over $q\cup r$, to a sum over $f(a)$ and $|f(a)|=|N|$.
			We omit arguments for Condition~\ref{cond:greater sum} as they run parallel to Condition~\ref{cond:lesser sum}.
			
			To verify Condition~\ref{cond:well defined}, we consider each case of (\ref{eqn:nu}) separately.
			Equations~(\ref{eqn:nu:Q}), (\ref{eqn:nu:q}), (\ref{eqn:nu:sR}), (\ref{eqn:nu:rS}), and (\ref{eqn:nu:else}) are trivial as these probabilities are well defined by definition.
			Lemma~\ref{thm:the bbP property} provides:\vs\vs
			\begin{equation}\label{eqn:bbP equvalence}
             \bbP^\alpha_i(1)\geq\hat{\bbP}^a_i(1)\Leftrightarrow\hat{\bbP}^a_i(0)\geq\bbP^\alpha_i(0)\vs\vs
			\end{equation}
			where the right hand side follows from $\bbP_i(1,a',w)+\bbP_i(0,a',w)=1=\bbP_i(1,a',w_0)+\bbP_i(0,a',w_0)$.
			Equation~(\ref{eqn:nu:R}) follows directly from the hypothesis and (\ref{eqn:nu:r}) holds from the right side of the equivalence.
			
			The lone remaining case is (\ref{eqn:nu:corner}), for which we define
			sets $N_{q}=\{g(a,z)\mid z\in q\}$, $N_{Q}=\{g(\alpha^T,z)\mid z\in Q\}$ and so on for $r,s,R,S$.
			We denote unions of these sets as $N_{qr}:=N_q\cup N_r$,  $N_{QR}:=N_Q\cup N_R$ and so on for other combinations of $q,r,s$ and $Q,R,S$.
			Recalling $q,r,Q,R$ are partitions over states that $a,\alpha^T$ may transition to, similarly, $N_{qr}$, $N_{QR}$ are partitions of agents whose unilateral deviations result in such transitions.
			This enables us to expand (\ref{eqn:nu:corner}):
            \vspace{-1.5mm}
			\begin{equation}
			\begin{aligned} \label{eqn:nu a alpha}
				\nu^{a,\alpha}(\bar{a},\bar{a}')&=
				\frac{1}{|N|}\Bigg(\sum_{i\in N_{qr}\cap N_{QR}}(1-\hat{\bbP}^a_i(\neg a_i)\\
				&\qquad-\bbP^\alpha_i(\neg \alpha^T_i))\\
				&\qquad+\sum_{i\in N_{qr}\setminus N_{QR}}(1-\hat{\bbP}^a_i(\neg a_{i}))\\
				&\qquad+\sum_{i\in N_{QR}\setminus N_{qr}}(1-\bbP^\alpha_i(\neg \alpha^T_{i}))\Bigg).
			\end{aligned}\vs\vs		
			\end{equation}
			This expansion takes advantage of $|N|=|f(a)|$ which allows $|N|$ to enter the sums as $1$.
			It now suffices to show that the summand of each sum is a well defined probability, of which the last two terms clearly are.
			
		    We begin by investigating $i\in N_{qr}\cap N_{QR}$.
		    In particular, we have $N_q=N_Q,N_s=N_R,N_r=N_S$ due to $b^{a,\alpha^T}$ and its bijectiveness due to Lemmas~\ref{thm:bijection}.
		    By disjointness of $q,r$ we have $N_{qr}=N_{QS}$ which we apply to $N_{qr}\cap N_{QR}=N_{QS}\cap N_{QR}=N_Q=N_q$.
			Applying definitions of $q,Q$ we find $i\in N_q\implies\neg a_i=1,\neg \alpha^T_i=0$.
			Thus the summand of the first sum for $i\in N_q$ is given by\vs\vs
			\begin{equation}
			    1-\hat{\bbP}^a_i(1)-\bbP^\alpha_i(0)\geq 1-\bbP^\alpha_i(1)-\bbP^\alpha_i(0)=0\vs\vs
			\end{equation}
			wherein the inequality is by \eqref{eqn:bbP equvalence}, giving that the summands in the first term of~\eqref{eqn:nu a alpha} are themselves well defined probabilities.
			As all conditions have been met, $\nu^{a,\alpha}$ is a monotone coupling as desired.
			%
		   %
		%
			%
	\end{proof}
	
			\begin{figure*}
		\begin{subnumcases}{\nu^{a,\alpha}(\bar{a},\bar{a}')=\label{eqn:nu}}
	   				\frac{1}{|N|}\left(\bbP^\alpha_{g(\alpha^T,\bar{a}')}(1)- \hat{\bbP}^a_{g(a',\bar{a}')}(1)\right) &  $\bar{a}=a,\bar{a}'\in R$ \label{eqn:nu:R}\\
			\frac{1}{|N|}\bbP^\alpha_{g(\alpha^T,\bar{a}')}(\bar{a}'_{g(\alpha^T,\bar{a}')}) &  $\bar{a}=a,\bar{a}'\in Q$\label{eqn:nu:Q} \\
			\frac{1}{|N|}\left(\hat{\bbP}^a_{g(a,\bar{a})}(0)- \bbP^\alpha_{g(a,\bar{a})}(0)\right) &  $\bar{a}\in r,\bar{a}'=\alpha^T$ \label{eqn:nu:r} \\						
			\frac{1}{|N|}\hat{\bbP}^a_{g(a,\bar{a})}(\bar{a}_{g(a,\bar{a})}) &  $\bar{a}\in q,\alpha^T=\bar{a}'$ \label{eqn:nu:q}\\						
			\frac{1}{|N|}\hat{\bbP}^a_{g(a,\bar{a})}(1) & $\bar{a}=b^{a,\alpha^T}(\bar{a}'),\bar{a}'\in R$ \label{eqn:nu:sR}\\
			\frac{1}{|N|}\bbP^\alpha_{g(\alpha^T,\bar{a}')}(0) & $\bar{a}\in r,\bar{a}'=b^{a,\alpha^T}(\bar{a})$ \label{eqn:nu:rS}\\
			\begin{array}{r}
			\frac{1}{|N|}\Big(|N|-\sum\limits_{z\in q\cup r}\hat{\bbP}^a_{g(a,z)}(z_{g(a,z)}) 
			-\sum\limits_{z'\in Q\cup R}\bbP^\alpha_{g(\alpha^T,z')}(z'_{g(\alpha^T,z')})\Big)
			\end{array}
			& $a=\bar{a},\alpha^T=\bar{a}'$	\label{eqn:nu:corner}		\\			
					0 & \mbox{otherwise.} \label{eqn:nu:else}
		\end{subnumcases}
		\caption{The full specification of the one-step monotone coupling for Theorem~\ref{thm:nu}.
		We adopt the notational convention that $q,s,Q,S$ are assumed to take arguments $a,a'$ and $r,R$ take the argument $a,a'$ respectively.
		\vspace{-5mm}
		\label{fig:onestep coupling}}
		\end{figure*}
	
\subsection{A monotone coupling over histories}
    We now present coupling $\nu^{\hat{g}}_\pi$ which is constructed using the one-step coupling.
    Using this coupling we then go on to prove Lemma~\ref{thm:proof of stoch dom+}.
    We define indicator function $\mathds{1}$ such that $\mathds{1}(P)=1$ if $P$ is a true logical proposition and $\mathds{1}(P)=0$ else.
	\begin{theorem}\label{thm:paths}
		
		Let $g\in\mathcal{G}^{\cal A}$ be an aligned history-dependent game and $\hat{g}$ be its corresponding exact potential game.
		Then $\nu_\pi^{\hat{g}}:\mathcal{A}_T^2\rightarrow[0,1]$ is a monotone coupling between $\hat{P}_\pi,P_\pi$.
		This coupling is given by\vs\vs\vs
		\begin{equation}\label{eqn:path coupling}
		\begin{aligned}
			\nu^{\hat{g}}_\pi(\alpha,  \bar{\alpha})=
			\pi(\alpha^1)\mathds{1}(\alpha^1=\bar{\alpha}^1)\prod_{t=1}^{T-1}\nu^{\alpha^t,\bar{\alpha}^{\leq t}}(\alpha^{t+1},\bar{\alpha}^{t+1})
		\end{aligned}\vs\vs\vs
		\end{equation}
		where $\alpha,\bar{\alpha}\in\mathcal{A}_T$, $\pi\in \Delta(A)$.
	\end{theorem}
	\begin{proof}
		Let $\alpha,\bar{\alpha}\in \mathcal{A}_T$ and let $g\in\mathcal{G}^\mathcal{A}$,  and let $\hat{g}$ be the corresponding exact potential game.
		We begin by showing that if $\alpha\nleq_{\mathcal{A}_T}\bar{\alpha}$, then $\nu^{\hat{g}}_\pi(\alpha,\bar{\alpha})=0$.
		Immediately, we have $\nu^{\hat{g}}_\pi(\alpha,\bar{\alpha})=0$ if $\alpha^1\neq\bar{\alpha}^1$, so we need only consider cases where $\alpha^1=\bar{\alpha}^1$.
		Inductively we find that if $\alpha\nleq_{\mathcal{A}_T}\bar{\alpha}$ there must exist some $t\in\{1,2,3,\dots,T-1\}$ such that $\alpha^t\leq_{A}\bar{\alpha}^t$ but $\alpha^{t+1}\nleq_{A}\bar{\alpha}^{t+1}$, and let $t$ be the minimal such value.
		In this case we have $\nu^{\alpha^t,\bar{\alpha}^{\leq t}}(\alpha^{t+1},\bar{\alpha}^{t+1})=0$ because $\nu^{\alpha^t,\hat{\alpha}^{\leq t}}$ is a well defined monotone coupling by Theorem~\ref{thm:nu}, yielding $\nu^{\hat{g}}_\pi(\alpha,\bar{\alpha})=0$ as desired.
		It also follows that $\nu^{\hat{g}}_\pi$ will always yield a well defined probability as it is either $0$ or a product of well defined probabilities.
		Thus we only need to show that the marginal probabilities are preserved given by \eqref{eqn:monotonic coupling}. 
		We begin by showing the left equation of \eqref{eqn:monotonic coupling}, that is:\vs
		\begin{equation}\label{eqn:marginal path coupling}
			\sum_{\alpha\leq_{\mathcal{A}_T}z}\nu^{\hat{g}}_\pi(\alpha,z)=\hat{P}_\pi(\alpha) \mbox{ for each }z\in\mathcal{A}_T\vs
		\end{equation}
		and omit the proof for the right hand equation as it proceeds identically.
		By inspecting (\ref{eqn:path coupling}), we only need to consider $z$ such that $z^1=\alpha^1$ and $z$ features at most a single unilateral deviation between any $t,t+1$.
		With these two conditions we rewrite
		\vs\vs
		\begin{equation}\label{eqn:combitorial form}
		\begin{aligned}
			\sum_{\alpha\leq_{\mathcal{A}_T}z}\nu^{\hat{g}}_\pi(\alpha,z) &=\sum_{\alpha\leq_{\mathcal{A}_T}z}\pi(\alpha^1)\prod_{t=1}^{T-1}\nu^{\alpha^t,z^{\leq t}}(\alpha^{t+1},z^{t+1})\\
			&=\pi(\alpha^1)\sum_{\alpha^2\leq_A z^2}\hs[2]\nu^{\alpha^1,z^{\leq 1}}(\alpha^{2},z^{2})\ldots\\
			\quad\quad&\sum_{\alpha^{T}\leq_A z^{T}}\hs[2]\nu^{\alpha^{T-1},z^{\leq T-1}}(\alpha^{T},z^{T}).
		\end{aligned}\vs
		\end{equation}
	as the combinatorial form. 
	Critically, this allows us to to apply the marginal sum properties of $\nu^{\alpha^t,z^{\leq t}}$ from Theorem~\ref{thm:nu} for each $t\in\{1,2,..,T\}$.
	First, considering the rightmost sum in \eqref{eqn:combitorial form}, it holds that\vs
	\begin{equation}
	\begin{aligned}
		&\sum_{\alpha^{T}\leq_A z^{T}}\nu^{\alpha^{T-1},z^{\leq T-1}}(\alpha^{T},z^{T})=\hat{P}^{\alpha^{T-1}}(\alpha^T).
	\end{aligned}\vs
	\end{equation}
	Because this has no dependence on $z$ we may factor out $\hat{P}^{\alpha^{T-1}}(\alpha^T)$ and repeat the process on the new rightmost sum.
	After performing this process recursively on all sums, we have \vs\vs
	\begin{equation}
	\sum_{\alpha\leq_{\mathcal{A}_T}z}\nu^{\hat{g}}_\pi(\alpha,z)=\pi(\alpha^1)\prod_{t=1}^{T-1}\hat{P}^{\alpha^{t}}(\alpha^{t+1})=\\
	\hat{P}_\pi(\alpha)\vs\vs
	\end{equation}
	as desired, noting we accounted for the indicator functions in $\nu^{\hat{g}}_\pi$.
	This concludes the proof of Theorem~\ref{thm:paths}.
	\end{proof}

Now that the necessary results have been developed we proceed with the proof of Lemma~\ref{thm:proof of stoch dom+}.
\subsubsection*{Proof of Lemma~\ref{thm:proof of stoch dom+}}
        Let $g\in\mathcal{G}^{\mathcal{A}}$ be an aligned history-dependent game and $\mathcal{I}\subset\mathcal{A}_T$ be an upper set.
        Define $\mathds{1}_\mathcal{I}(\alpha):=\mathds{1}(\alpha\in\mathcal{I})$ as an indicator function.
    	Consider probability measures $P_\pi,\hat{P}_\pi$ coupled by $\nu^{\hat{g}}_\pi$ in Theorem~\ref{thm:paths}, we have\vspace{-1.5mm}
    	\begin{equation}\label{eqn:stoch dom}
    	\begin{aligned}
            P_\pi(\mathcal{I})-\hat{P}_\pi(\mathcal{I})&=\mathbb{E}_{P_\pi}(\mathds{1}_{\mathcal{I}})-\mathbb{E}_{\hat{P}_\pi}(\mathds{1}_{\mathcal{I}}) \\
            &=\nu^{\hat{g}}_\pi(\mathcal{I^C},\mathcal{I})\geq0.
    	\end{aligned}\vspace{-1.5mm}
    	\end{equation}
    	where the second equality follows by Proposition~\ref{thm:increasing functions in expectation} as $\mathds{1}_{\mathcal{I}}$ is increasing in $\mathcal{A}_T$.
    	Note \eqref{eqn:stoch dom} runs parallel to the proof of \cite[Corollary 3]{paarporn2017}.
		That is, for any upper set $\mathcal{I}\subset \mathcal{A}_T$ we have \vs\vs
		\begin{equation} \label{eqn:stoch dom of Ppi hatPpi}
		    P_\pi(\mathcal{I})\geq \hat{P}_\pi(\mathcal{I}).\vs\vs
		\end{equation}
		Let $((\vec{0})^{T-1}_{t=1},\vec{1})\in \mathcal{I}$.
		This induces $\mathcal{I}$ such that it includes every path such that at time $T$ the $\vec{1}$ state is played.
		This yields the following interpretation \vs
		\begin{equation}\label{eqn:probability of 1}
		P_\pi(\mathcal{I})=\mbox{Pr}(s(T;\tau,\pi,g)=\vec{1})\vs
		\end{equation}
		representing the probability that at time $T$ game $g$ is in the $\vec{1}$ action profile given initial distribution $\pi\in\Delta(A)$ and learning temperature parameter $\tau$.
		Noting a parallel interpretation to \eqref{eqn:probability of 1} holds for $\hat{P}_\pi,\hat{g}$, we apply these to \eqref{eqn:stoch dom of Ppi hatPpi} to obtain\vs\vs
		\begin{equation}
		    \mbox{Pr}(s(T;\tau,\pi,g)=\vec{1})\geq\mbox{Pr}(s(T;\tau,\pi,\hat{g})=\vec{1})\vs\vs
		\end{equation}
		as desired.
    \hfill$\blacksquare$

\vs\vs
\bibliographystyle{ieeetr}
\bibliography{library}

\end{document}